\documentclass[11pt]{article}%
\usepackage[cm]{fullpage}
\usepackage[usenames,dvipsnames,svgnames,table]{xcolor}%
\usepackage{bm}
\usepackage{paralist}%
\usepackage{mleftright}%
\usepackage{stmaryrd}%

\usepackage{xspace}%
\usepackage{ifpdf}%
\usepackage{graphicx}%
\usepackage[normalem]{ulem} %
\usepackage{caption}%
\usepackage{amsmath}%
\usepackage{amssymb}%
\usepackage[amsmath,thmmarks]{ntheorem}%
\theoremseparator{.}%

\usepackage{euscript}
\usepackage{mathrsfs}

\usepackage{verbatim}%

\usepackage{titlesec}%
\titlelabel{\thetitle. }%

\usepackage{bbm}%
\usepackage{bbding}

\usepackage{hyperref}%
\ifx\PDFBlackWhite\undefined%
\hypersetup{%
  breaklinks,%
  ocgcolorlinks, %
  colorlinks=true,%
  urlcolor=[rgb]{0.45,0.0,0.0},%
  linkcolor=[rgb]{0.0,0.45,0.0},%
  citecolor=[rgb]{0,0,0.45}%
}%

\fi

\allowdisplaybreaks

\ifx\UseBibLatex\undefined%
  \newcommand{\BibLatexMode}[1]{}
  \newcommand{\BibTexMode}[1]{#1}
\else
  \newcommand{\BibLatexMode}[1]{#1}
  \newcommand{\BibTexMode}[1]{}
\fi

\newcommand{\eps}{\Mh{\varepsilon}}%

\newlength{\savedparindent}

\newcommand{\Term}[1]{\textsf{#1}}

\definecolor{blue25}{rgb}{0, 0, 11}
\newcommand{\emphic}[2]{%
  \textcolor{blue25}{%
    \textbf{\emph{#1}}}%
  \index{#2}}

\ifx\PDFBlackWhite\undefined
\else
\renewcommand{\emphic}[2]{\textbf{\emph{#1}}}
\fi

\newcommand{\emphi}[1]{\emphic{#1}{#1}}

\newcommand{\cardin}[1]{\left| {#1} \right|}%

\newcommand{\pth}[1]{\mleft({#1}\mright)}
\newcommand{\brc}[1]{\left\{ {#1} \right\}}

\newcommand{\Set}[2]{\left\{ #1 \;\middle\vert\; #2 \right\}}
\newcommand{\pbrc}[1]{\mleft[ {#1} \mright]}

\newcommand{\Ex}[1]{\mathop{\mathbf{E}}\pbrc{#1}}

\newcommand{\ProbY}[2]{{\mathbf{Pr}}^{}_{#1}\pbrc{#2}}

\newcommand{\remove}[1]{}

\newtheorem{theorem}{Theorem}[section]%
\newtheorem{lemma}[theorem]{Lemma}%
\newtheorem*{restate*}[theorem]{Restatement of }%
\newcommand{\myqedsymbol}{\rule{2mm}{2mm}}

\theoremstyle{remark}%
\theoremheaderfont{\sf}%
\theorembodyfont{\upshape}%
\newtheorem{defn}[theorem]{Definition}
\newtheorem*{remark:unnumbered}[theorem]{Remark}%

\theoremheaderfont{\em}%
\theorembodyfont{\upshape}%
\theoremstyle{nonumberplain}%
\theoremseparator{}%
\theoremsymbol{\myqedsymbol}%
\newtheorem{proof}{Proof:}%

\numberwithin{figure}{section}%
\numberwithin{table}{section}%
\numberwithin{equation}{section}%

\newcommand{\HLinkPageSuffix}[3]{\hyperref[#2]{#1\ref*{#2}%
      #3$_\text{p\pageref{#2}}$}}%
\newcommand{\HLinkSuffix}[3]{\hyperref[#2]{#1\ref*{#2}{#3}}}
\newcommand{\HLinkShort}[2]{\hyperref[#2]{#1\ref*{#2}}}
\newcommand{\HLink}[2]{\hyperref[#2]{#1~\ref*{#2}}}
\newcommand{\HLinkPage}[2]{\hyperref[#2]{#1~\ref*{#2}%
    $_\text{p\pageref{#2}}$}}

\newcommand{\eqrefpar}[1]{\hyperref[equation:#1]{(\ref*{equation:#1})}} %

\providecommand{\deflab}[1]{\label{def:#1}}

\newcommand{\lemlab}[1]{\label{lemma:#1}}
\newcommand{\lemref}[1]{\HLink{Lemma}{lemma:#1}}

\newcommand{\pr}{\Mh{\varphi}}%

\providecommand{\Mh}[1]{{#1}}

\renewcommand{\th}{th\xspace}

\renewcommand{\Re}{{\mathbb{R}}}
\newcommand{\PntSet}{\ensuremath{\Mh{P}}\xspace}%
\newcommand{\PntSetA}{\ensuremath{\Mh{Q}}\xspace}%
\newcommand{\PSetFar}{\ensuremath{\PntSet^{}_{\geq}\xspace}}%

\newcommand{\cTimes}{\Mh{\beta}}%

\newcommand{\DA}{\Mh{D}}%
\newcommand{\DSTimes}{\Mh{L}}%

\newcommand{\rad}{\Mh{r}}%

\newcommand{\subseq}{\Mh{{{m}}}}%

\newcommand{\subseqA}{\Mh{{{u}}}}%

\newcommand{\seq}{{\Mh{M}}}%

\newcommand{\pnt}{\Mh{\mathsf{x}}}%
\newcommand{\pntc}{\Mh{{x}}}%

\newcommand{\pntA}{\Mh{\mathsf{v}}}%

\newcommand{\tTimes}{\Mh{T}}%

\newcommand{\query}{\Mh{\mathsf{q}}}%

\newcommand{\norm}[2]{\left\| {#2} \right\|_{#1}}

\newcommand{\atgen}{\symbol{'100}}

\newcommand{\SarielThanks}[1]{%
   \thanks{%
      Department of Computer Science; %
      University of Illinois; %
      201 N. Goodwin Avenue; %
      Urbana, IL, %
      61801, USA; %
      {\tt \si{sariel}\atgen{}\si{illinois.edu}}; %
      {\tt \url{http://sarielhp.org}.} #1 } 
}

\newcommand{\SepidehThanks}[1]{%
   \thanks{%
      Department of EECS; 
      MIT; %
      77 Massachusetts Avenue, Cambridge, MA 02139, USA;
      {\tt mahabadi@mit.edu}. %
      #1 
   }%
}

\newcommand{\si}[1]{#1}%
\newcommand{\IntRange}[1]{\left\llbracket #1 \right\rrbracket}
\newcommand{\DistD}[1]{\Mh{\mathcal{D}}_{#1}}

\DefineNamedColor{named}{OliveGreen}{cmyk}{0.64, 0, 0.95, 0.40}
\definecolor{OliveGreen}{cmyk}{0.64, 0, 0.95, 0.40 }

\newcommand{\LSH}{\Term{LSH}\xspace}%

\IfFileExists{sariel_computer.sty}{\def\sarielComp{1}}{}
\ifx\sarielComp\undefined%
\newcommand{\SarielComp}[1]{}
\newcommand{\NotSarielComp}[1]{#1}%
 \else
\newcommand{\SarielComp}[1]{#1}%
\newcommand{\NotSarielComp}[1]{}%
\fi

\SarielComp{%
   \ifx\colorMath\undefined%
   \else
   \DefineNamedColor{named}{ColorMath}{rgb}{0.5,0.2,0}
       \renewcommand{\Mh}[1]{{\textcolor{ColorMath}{#1}}}%
   \fi
}

\BibLatexMode{%

\usepackage[style=alphabetic,backend=biber]{biblatex}

\ExecuteBibliographyOptions{%
      maxnames      = 40,
      maxalphanames = 4,
      maxcitenames  = 3,
      minalphanames = 3,
      minnames      = 1,
      firstinits    = true, 
      doi           = false,
      url           = false,
      isbn          = false
    }

\renewbibmacro{in:}{}
\DeclareFieldFormat[inproceedings]{date}{}%
\DeclareFieldFormat[inproceedings]{pages}{#1, \printfield{year}}
\DeclareFieldFormat[inproceedings]{booktitle}{\mkbibemph{#1,\nopunct}}
\AtEveryBibitem{%
\ifentrytype{inproceedings}{%
   \clearlist{publisher}%
   \clearlist{location}%
   \clearfield{series}%
}{}
}

\ExecuteBibliographyOptions{doi=false}
\ExecuteBibliographyOptions{url=false}
\newbibmacro{string+doi}[1]{%
   \iffieldundef{doi}{%
      \iffieldundef{url}{%
         #1%
      }{%
         \href{\thefield{url}}{#1}%
      }%
   }{%
      \href{http://dx.doi.org/\thefield{doi}}{#1}}%
}%
\DeclareFieldFormat*{title}{\usebibmacro{string+doi}{\mkbibemph{#1}}}

}

\BibLatexMode{%
}

\title{\LSH on the Hypercube Revisited}

\author{%
   Sariel Har-Peled%
   \SarielThanks{%
      Work on this paper was partially supported by a NSF AF awards
      CCF-0915984 and CCF-1217462.}%
   \and%
   Sepideh Mahabadi%
   \SepidehThanks{}%
}

\begin{document}

\maketitle

\begin{abstract}
    \LSH (locality sensitive hashing) had emerged as a powerful
    technique in nearest-neighbor search in high dimensions
    \cite{im-anntr-98, him-anntr-12}. Given a point set $P$ in a
    metric space, and given parameters $\rad$ and $\eps > 0$, the task
    is to preprocess the point set, such that given a query point
    $\query$, one can quickly decide if $\query$ is in distance at
    most $\leq \rad$ or $\geq (1+\eps)\rad$ from the query point. Once
    such a near-neighbor data-structure is available, one can reduce
    the general nearest-neighbor search to logarithmic number of
    queries in such structures \cite{h-rvdnl-01,
       im-anntr-98,him-anntr-12}.

    In this note, we revisit the most basic settings, where $P$ is a
    set of points in the binary hypercube $\brc{0,1}^d$, under the
    $L_1$/Hamming metric, and present a short description of the \LSH
    scheme in this case. We emphasize that there is no new
    contribution in this note, except (maybe) the presentation itself,
    which is inspired by the authors recent work \cite{hm-padra-17}.
\end{abstract}

\section{Locality sensitive hashing revisited}

\subsection{Preliminaries}

\begin{defn}
    Consider a sequence $\subseq$ of $k$, not necessarily distinct,
    integers $i_1, i_2, \ldots, i_k \in \IntRange{d}$, where
    $\IntRange{d} = \brc{1, \ldots, d}$.  For a point
    $\pnt = (\pntc_1, \ldots, \pntc_d) \in \Re^d$, its
    \emphi{projection} by $\subseq$, denoted by $\subseq \pnt$ is the
    point $\pth{ \pntc_{i_1}, \ldots, \pntc_{i_k}} \in \Re^k$.
    Similarly, the \emph{projection} of a point set
    $\PntSet \subseteq \Re^d$ by $\subseq$ is the point set
    $\subseq \PntSet = \Set{\subseq \pnt}{\pnt \in \PntSet}$.
\end{defn}

Given two sequences $\subseq = i^{}_1, \ldots, i^{}_k$ and
$\subseqA = j^{}_1, \ldots, j^{}_{k'}$, let $\subseq | \subseqA$
denote the \emph{concatenated} sequence
\begin{math}
    \subseq | \subseqA = i^{}_1, \ldots, i^{}_k, j^{}_1, \ldots,
    j^{}_{k'}.
\end{math}
Given a probability $\pr$, a natural way to create such a projection,
is to include the $i$\th coordinate, for $i=1,\ldots, d$, with
probability $\pr$. Let $\DistD{\pr}$ denote the distribution of such
sequences of indices.

\begin{defn}
    \deflab{t:splay}%
    Let $\DistD{\pr}^\tTimes$ denote the distribution resulting from
    concatenating $t$ independent sequences sampled from
    $\DistD{\pr}$. The length of a sampled sequence is
    $d \tTimes$.
\end{defn}

Observe that for a point $\pnt \in \brc{0,1}^d$, and
$\seq \in \DistD{\pr}^\tTimes$, the projection $\seq \pnt$ might be
higher dimensional than the original point $\pnt$, as it might contain
repeated coordinates of the original point.

\subsection{Algorithm}

\paragraph{Input.}
The input is a set $\PntSet$ of $n$ points in the hypercube
$\brc{0,1}^d$, and parameters $\rad$ and $\eps$.

\paragraph{Preprocessing.}
We set parameters as follows:
\begin{align*}
  \cTimes = \frac{1}{1+\eps} \in (0,1),%
  \quad  %
  \pr = 1 - \exp \pth{ - \frac{1}{\rad}} \approx \frac{1}{\rad},
  \quad %
  \tTimes = \cTimes \ln n,
  \quad
  \text{and}%
  \quad
  \DSTimes = O( n^{\cTimes} \log n).
\end{align*}
We randomly and independently pick $\DSTimes$ sequences
\begin{math}
    \seq_1, \ldots, \seq_\DSTimes \in \DistD{\pr}^\tTimes.
\end{math}
Next, the algorithm computes the point sets
$\PntSetA_i = \seq_i \PntSet_i$, for $i=1,\ldots, \DSTimes$, and stores
them each in a hash table, denoted by $\DA_i$, for
$i=1,\ldots, \DSTimes$.

\paragraph{Answering a query.}

Given a query point $\query \in \brc{0,1}^d$, the algorithm computes
$\query_i = \seq_i \query$, for $i=1,\ldots, \DSTimes$.  From each
$\DA_i$, the algorithm retrieves a list $\ell_i$ of all the points
that collide with $\query_i$. The algorithm scans the points in the
lists $\ell_1, \ldots, \ell_\DSTimes$. If any of these points is in
Hamming distance smaller than $(1+\eps)\rad$, the algorithm returns it
as the desired near-neighbor (and stops). Otherwise, the algorithm
returns that all the points in $\PntSet$ are in distance at least
$\rad$ from $\query$.

\subsection{Analysis}

\subsubsection{Correctness}

\begin{lemma}
    \lemlab{success:simple}%
    Let $K$ be a set of $\rad$ marked/forbidden coordinates. The
    probability that a sequence
    $\seq = (\subseq_1,\ldots, \subseq_\tTimes)$ sampled from
    $\DistD{\pr}^\tTimes$ does not sample any of the coordinates of
    $K$ is
    \begin{math}
        1/n^{\cTimes}.
    \end{math}
    This probability increases if $K$ contains fewer coordinates.
\end{lemma}

\begin{proof}
    For any $i$, the probability that $\subseq_i$ does not contain any
    of these coordinates is
    \begin{math}
        (1-\pr)^\rad%
        =%
        \pth{e^{-1/\rad }}^{\rad}%
        =%
        1/e.
    \end{math}
    Since this experiment is repeated $\tTimes$ times, the probability
    is
    \begin{math}
        e^{-\tTimes}%
        =%
        e^{ - \cTimes \ln n}%
        =%
        n^{- \cTimes}.
    \end{math}
    \qquad
\end{proof}

\begin{lemma}
    \lemlab{lists}%
    We
    have the following:
    \begin{compactenum}[\;(A)]
        \item Let $\pnt$ be the nearest-neighbor to $\query$ in
        $\PntSet$.  If $\norm{1}{\query -\pnt} \leq \rad$ then, with
        high probability, the data-structure returns a point that is
        in distance $\leq (1+\eps)\rad$ from $\query$.

        \item In expectation, the total number of points in
        $\ell_1,\ldots, \ell_\DSTimes$ that are in distance
        $\geq (1+\eps)\rad$ from $\query$ is $\leq \DSTimes$.
    \end{compactenum}
\end{lemma}

\begin{proof}
    (A) The good event here is that $\pnt$ and $\query$ collide under
    one of the sequences of $\seq_1,\ldots, \seq_\DSTimes$.  However,
    the probability that $\seq_i \pnt = \seq_i \query$ is at least
    $1/n^{\cTimes}$, by \lemref{success:simple}, as this is the
    probability that $\seq_i$ does not sample any of the (at most
    $\rad$) coordinates where $\pnt$ and $\query$ disagree. As such,
    the probability that all $\DSTimes$ data-structures fail (i.e.,
    none of the lists $\ell_1,\ldots, \ell_\DSTimes$ contains $\pnt$),
    is at most $(1- 1/n^{\cTimes})^{\DSTimes} < 1/n^{O(1)}$, as
    $\DSTimes = O\pth{ n^{\cTimes} \log n}$.

    (B) Let $\PSetFar$ be the set of points in $\PntSet$ that are in
    distance $\geq (1+\eps)\rad$ from $\query$.  For a point
    $\pntA \in \PSetFar$, with $\Delta = \norm{1}{\pntA - \query}$, we
    have that the probability for $\seq \in \DistD{\pr}^\tTimes$
    misses all the $\Delta$ coordinates, where $\pntA$ and $\query$
    differ, is
    \begin{align*}
      (1-\pr)^{\Delta}%
      \leq 
      (1-\pr)^{(1+\eps)\rad \tTimes}%
      =%
      \pth{e^{-1/r}}^{(1+\eps)\rad \tTimes}%
      =%
      \exp\pth{ -(1+\eps) \cTimes \ln n }%
      =%
      \frac{1}{n},
    \end{align*}
    as $\pr = 1-e^{-1/\rad}$, $\tTimes = \cTimes \ln n$, and
    $\cTimes = 1/(1+\eps)$. But then, for any $i$, we have
    \begin{align*}
      \Ex{\bigl.\cardin{\ell_i}} = \sum_{\pnt \in \PSetFar}
      \ProbY{\bigl.\seq_i}{\seq_i \pnt = \seq_i \query }
      \leq%
      \cardin{\PSetFar} \frac{1}{n}
      \leq%
      1.
    \end{align*}
    As such, the total number of far points in the lists is at most
    $\DSTimes \cdot 1 = \DSTimes$, implying the claim.
\end{proof}

\subsubsection{Running time}

For each $i$, the query computes $\seq_i \query$ and that takes
$O(d \tTimes) = O(d\log n)$ time. Repeated $\DSTimes$ times, this takes
$O( \DSTimes d \log n)$ time overall.  Let $X$ be the random variable
that is the number of points in the extracted lists that are in
distance $\geq (1+\eps)\rad$ from the query point.  The time to scan
the lists is
\begin{math}
    O\pth{\bigl. d( X +1) },
\end{math}
 since the algorithm stops as
soon as it finds a near point. As such, by \lemref{lists} (B), the
expected query time is
\begin{math}
    O( \DSTimes d \log n + \DSTimes d) = O\pth{d n^{1/(1+\eps)} \log^2
       n}.
\end{math}

\newcommand{\uniqX}[1]{\Mh{\mathrm{uniq}}\pth{#1}}%

\subsubsection{Improving the performance (a bit)}

Observe that for $\seq \in \DistD{\pr}^\tTimes$, and any two points
$\pnt, \pntA \in \brc{0,1}^d$, all the algorithm cares about is
whether $\seq \pnt = \seq \pntA$. As such, if a coordinate is probed
many times by $\seq$, we might as well probe this coordinate only
once. In particular, for a sequence $\seq \in \DistD{\pr}^\tTimes$,
let $\seq' = \uniqX{\seq}$ be the projection sequence resulting from
removing replications in $\seq$. Significantly, $\seq'$ is only of
length $\leq d$, and as such, computing $\seq' \pnt$, for a point
$\pnt$, takes only $O(d)$ time. It is not hard to verify that one can
also sample directly $\uniqX{\seq}$, for
$\seq \in \DistD{\pr}^\tTimes$, in $O(d)$ time. This improves the
query and processing by a logarithmic factor.

\subsubsection{Result}

\begin{theorem}
    Given a set $P$ of $n$ points in $\brc{0,1}^d$, and parameters
    $\rad, \eps$, one can preprocess $P$ in
    $O( d n^{1+1/(1+\eps)} \log n)$ time and space, such that given a
    query point $\query$, the algorithm returns, in expected
    $O( d n^{1/(1+\eps)} \log n)$ time, one of the following:
    \smallskip%
    \begin{compactenum}[\quad(A)]
        \item a point $\pnt \in P$ such that
        $\norm{1}{\query - \pnt} \leq (1+\eps)\rad$, or

        \item the distance of $\query$ from $P$ is larger than $\rad$.
    \end{compactenum}
    \smallskip%
    The algorithm may return either result if the distance of $\query$
    from $P$ is in the range $[r,(1+\eps)r]$. The algorithm succeeds
    with high probability (per query).
\end{theorem}

One can also get a high-probability guarantee on the query time. For a
parameter $\delta > 0$, create $O( \log \delta^{-1})$ \LSH
data-structures as above. Perform the query as above, except that when
the query time exceeds (say) twice the expected time, move on to redo
the query in the next \LSH data-structure. The probability that the
query had failed on one of these \LSH data-structures is
$\leq 1/2$, by Markov's inequality. As such, overall, the query time
becomes $O( d n^{1/(1+\eps)} \log n \log \delta^{-1})$, with
probability $\geq 1-\delta$.

 
\hypersetup{%
}

\BibTexMode{%
 \providecommand{\CNFX}[1]{ {\em{\textrm{(#1)}}}}
  \providecommand{\tildegen}{{\protect\raisebox{-0.1cm}{\symbol{'176}\hspace{-0.03cm}}}}
  \providecommand{\SarielWWWPapersAddr}{http://sarielhp.org/p/}
  \providecommand{\SarielWWWPapers}{http://sarielhp.org/p/}
  \providecommand{\urlSarielPaper}[1]{\href{\SarielWWWPapersAddr/#1}{\SarielWWWPapers{}/#1}}
  \providecommand{\Badoiu}{B\u{a}doiu}
  \providecommand{\Barany}{B{\'a}r{\'a}ny}
  \providecommand{\Bronimman}{Br{\"o}nnimann}  \providecommand{\Erdos}{Erd{\H
  o}s}  \providecommand{\Gartner}{G{\"a}rtner}
  \providecommand{\Matousek}{Matou{\v s}ek}
  \providecommand{\Merigot}{M{\'{}e}rigot}
  \providecommand{\Hastad}{H\r{a}stad\xspace}
  \providecommand{\CNFCCCG}{\CNFX{CCCG}}
  \providecommand{\CNFBROADNETS}{\CNFX{BROADNETS}}
  \providecommand{\CNFESA}{\CNFX{ESA}}
  \providecommand{\CNFFSTTCS}{\CNFX{FSTTCS}}
  \providecommand{\CNFIJCAI}{\CNFX{IJCAI}}
  \providecommand{\CNFINFOCOM}{\CNFX{INFOCOM}}
  \providecommand{\CNFIPCO}{\CNFX{IPCO}}
  \providecommand{\CNFISAAC}{\CNFX{ISAAC}}
  \providecommand{\CNFLICS}{\CNFX{LICS}}
  \providecommand{\CNFPODS}{\CNFX{PODS}}
  \providecommand{\CNFSWAT}{\CNFX{SWAT}}
  \providecommand{\CNFWADS}{\CNFX{WADS}}

}

\BibLatexMode{\printbibliography}

\begin{thebibliography}{{Har}01}

\bibitem[{Har}01]{h-rvdnl-01}
\href{http://sarielhp.org}{S.~{{Har-Peled}}}.
\newblock \href{http://sarielhp.org/p/01/avoronoi}{A replacement for {Voronoi}
  diagrams of near linear size}.
\newblock In {\em Proc. 42nd Annu. IEEE Sympos. Found. Comput. Sci.
  {\em(FOCS)}}, pages 94--103, 2001.

\bibitem[HIM12]{him-anntr-12}
\href{http://sarielhp.org}{S.~{{Har-Peled}}}, \href{http://theory.lcs.mit.edu/~indyk/}{P.~{Indyk}}, and R.~Motwani.
\newblock  Approximate nearest neighbors: {Towards} removing the curse of
  dimensionality.
\newblock {\em Theory Comput.}, 8:321--350, 2012.
\newblock Special issue in honor of Rajeev Motwani.

\bibitem[HM17]{hm-padra-17}
Sariel Har{-}Peled and Sepideh Mahabadi.
\newblock
  \href{http://epubs.siam.org/doi/abs/10.1137/1.9781611974782.1}{Proximity in
  the age of distraction: Robust approximate nearest neighbor search}.
\newblock In {\em Proc. 28th ACM-SIAM Sympos. Discrete Algs. {\em(SODA)}},
  pages 1--15, 2017.

\bibitem[IM98]{im-anntr-98}
\href{http://theory.lcs.mit.edu/~indyk/}{P.~{Indyk}} and R.~Motwani.
\newblock  Approximate nearest neighbors: {Towards} removing the curse of
  dimensionality.
\newblock In {\em Proc. 30th Annu. ACM Sympos. Theory Comput. {\em(STOC)}},
  pages 604--613, 1998.

\end{thebibliography}


\end{document}
